\documentclass[%
 aip,
 jmp,%
 amsmath,amssymb,
 preprint%
]{revtex4-1}

\usepackage{amsthm}
\newtheorem{theo}{Theorem}
\newtheorem{defi}{Definition}
\newtheorem{lemm}{Lemma}
\newtheorem{prop}{Proposition}

\usepackage{enumerate}
\usepackage{braket}
\usepackage{mathrsfs}

\usepackage{graphicx}
\usepackage{dcolumn}
\usepackage{bm}

\newcommand{\tr}{\mathrm{tr}}
\newcommand{\stsp}{\mathcal{S}(\mathcal{H})}

\begin{document}

\preprint{}

\title[]{
Minimal sufficient positive-operator valued measure on a separable Hilbert space
}

\author{Yui Kuramochi}
 \affiliation{Department of Nuclear Engineering, Kyoto University, 6158540 Kyoto, Japan}
 \email{kuramochi.yui.22c@st.kyoto-u.ac.jp}
 


\date{\today}

\begin{abstract}
We introduce a concept of a minimal sufficient positive-operator valued measure (POVM),
which is the least redundant POVM among the POVMs that have the equivalent information
about the measured quantum system.
Assuming the system Hilbert space to be separable,
we show that 
for a given POVM
a sufficient statistic called a Lehmann-Scheff\'{e}-Bahadur statistic
induces a minimal sufficient POVM.
We also show that 
every POVM has an equivalent minimal sufficient POVM and that
such a minimal sufficient POVM is unique up to 
relabeling neglecting null sets.
We apply these results to discrete POVMs and information conservation conditions proposed by the author.
\end{abstract}

\pacs{03.65.Ta, 03.67.-a, 02.30.Cj}
\keywords{quantum measurement, minimal sufficient positive-operator valued measure,
fuzzy equivalence relation}
\maketitle


\section{Introduction}
\label{sec:intro}
A measurement outcome statistic of a general quantum measurement process 
is described by 
a positive-operator valued measure~\cite{davieslewisBF01647093,davies1976quantum}
(POVM).
Such a description of quantum measurements
enables us to formulate  both projective and non-projective measurements.
However, for the price of such a generality,
some POVMs contain redundant information
irrelevant to the system
and, due to such redundancies,
for a given POVM 
there exist infinitely many POVMs that 
bring us the equivalent information about the system.

To clarify this point,
let us consider the following example of a pair of discrete POVMs
$A$ and $B$:
\begin{gather*}
	A = \{ A_0 , A_1 \} ,
	\\
	O \leq A_0 \leq I ,
	\,
	A_1 = I - A_0 ,
	\\
	B =
	\{
	B_{00} , B_{01} , B_{10} , B_{11}
	\} ,
	\\
	B_{i0} = \lambda A_i ,
	\,
	B_{i1} = (1 -\lambda ) A_{i}.
\end{gather*}
Here 
$\lambda \in (0,1)$ and
$A_i$ and $B_{ij}$ are bounded operators on a Hilbert space
$\mathcal{H}$
corresponding to the measured quantum system
and $O$ and $I$ are zero and identity operators,
respectively.
The measurement corresponding to the POVM $B$ can
be realized,
for example,
as follows:
perform the measurement $A$, whose measurement outcome is $i \in \{ 0,1 \}$,
and generate a binary random variable $j$, which gives 
$0$ with a probability $\lambda$
and $1$ with a probability $1 -\lambda ,$
and the measurement outcome of $B$ is given by a pair
$(i,j).$
Apparently $B = \{ B_{ij} \}$ contains redundant information,
which is in this case the classical random variable $j$,
and $A$ and $B$ give the equivalent information about the system.

Then it is natural to ask whether we can reduce such redundancies
for a given POVM
and how far such reductions proceed.
In this paper, 
to formulate and answer this question,
we introduce a concept of a minimal sufficient POVM
which corresponds to the least redundant POVM
among the POVMs that give the same information about the system.
The main finding of this paper
(Theorem~\ref{theo:exunique})
is that
for any POVM on a separable Hilbert space
there exists a minimal sufficient POVM
that has the equivalent information about the system
and that such a minimal sufficient POVM is unique up to 
almost isomorphism,
which is the relabeling neglecting null sets.

The concept of the minimal sufficient POVM
has two origins:
a minimal sufficient statistics~\cite{lehmann1950,bahadur1954}
in mathematical statistics
and a fuzzy equivalence 
relation~\cite{martens1990,Dorofeev1997349,Heinonen200577,Jencova2008}
known in quantum measurement theory.
Two POVMs are fuzzy equivalent if one of them can be  realized by
a classical post-processing of the other.
The concept of the fuzzy equivalence relation 
is used in the definition of the minimal sufficient POVM.
The minimal sufficiency condition for the POVM can be regarded as a
generalization of the minimal sufficient statistics
in the sense that
we consider a more general class of post-processing
which includes taking statistics.

This paper is organized as follows.
In Sec.~\ref{sec:preliminaries},
we show some preliminary results on mathematical statistics and quantum measurement theory
which will be used in Sec.~\ref{sec:minimal}.
In Sec.~\ref{sec:minimal},
we define
two definitions of the minimal sufficient POVM 
corresponding to two kinds of fuzzy equivalence relations
and introduce 
a sufficient statistic called a Lehmann-Scheff\'{e}-Bahadur 
(LSB) statistic
for a given POVM.
It is shown in
Theorem~\ref{theo:minimal_exists} that the
LSB statistic is a minimal sufficient statistic
for the given POVM
and a POVM induced by the LSB statistic is a minimal sufficient POVM.
After introducing a concept of almost isomorphism between POVMs,
we show 
in Theorem~\ref{theo:exunique}
that a POVM 
has an equivalent minimal sufficient POVM
and that such a minimal sufficient POVM is unique up to almost isomorphism.
In Sec~\ref{sec:discrete},
we consider
discrete POVMs
and prove that 
for a given discrete POVM
there is a unique equivalent minimal sufficient 
POVM that is discrete and has no zero elements.
In Sec.~\ref{sec:cons},
we apply the main results 
to information conservation conditions
proposed by the author~\cite{PhysRevA.91.032110,:/content/aip/journal/jmp/56/9/10.1063/1.4931625}.

\section{Preliminaries}
\label{sec:preliminaries}
In this section, we introduce some preliminary concepts and results 
on mathematical statistics and quantum measurement theory.

\subsection{Sufficient and minimal sufficient statistics}
Let $(\Omega , \mathscr{B})$ a measurable space.
A family of probability measures $\mathcal{P}$ with the outcome space $(\Omega , \mathscr{B})$
is called a statistical model on $(\Omega , \mathscr{B})$.

Let $\mathcal{P}$ be a statistical model on $(\Omega , \mathscr{B})$
and let $\lambda$ be a probability measure on $(\Omega , \mathscr{B})$.
$\mathcal{P}$ is said to be dominated by $\lambda$,
denoted by $\mathcal{P} \ll \lambda$,
if every element $P \in \mathcal{P}$ is absolutely continuous with respect to $\lambda .$
A statistical model $\mathcal{P}$ is said to be dominated if there exists a probability measure that dominates $\mathcal{P} .$
A dominated statistical model $\mathcal{P}$ has a countable subset $\{ P_i \}_{i \geq 1} \subset \mathcal{P}$
such that
$\lambda := \sum_{i \geq 1} c_i P_i$ dominates $\mathcal{P}$
if $c_i >0 $ and $\sum_{i\geq 1} c_i =1$ hold~\cite{halmos1949application}.
Such $\lambda$ is called a pivotal measure for $\mathcal{P}.$

Let $\mathcal{P} = \{ P_\theta \}_{\theta \in \Theta}$ be a statistical model
on an outcome space $(\Omega , \mathscr{B})$
and let $(\Omega_T , \mathscr{B}_T)$ be a measurable space.
A $\mathscr{B}/\mathscr{B}_T$-measurable map $T \colon \Omega \to \Omega_T$ 
is called a \textit{statistic}.
The set of $\mathscr{B}/ \mathscr{B}_T$-measurable maps (statistics) is denoted by
$\mathbb{M}( (\Omega , \mathscr{B}) \to (\Omega_T , \mathscr{B}_T) )$.
$T$ is said to be \textit{sufficient} with respect to the statistical model $\mathcal{P}$
if for every $E \in \mathscr{B}$, there exists a $\mathscr{B}_T$-measurable function
$P(E|\cdot ) \colon \Omega_T \to [0,\infty ) $ such that
$P_\theta (E|t) = P(E|t)$ $P^T_\theta $-a.e. 
for each $P_\theta \in \mathcal{P}$,
where $P_\theta (E|t)$ is the conditional probability of $P_\theta$ 
for given $T=t$
and $P_\theta^T (\cdot) := P_\theta (T^{-1} (\cdot)) .$

Let $P$ and $Q$ be probability measures with an outcome space $(\Omega, \mathscr{B})$.
Suppose $f \colon  [ 0 , \infty ) \to \mathbb{R}$ be a strictly convex function
such that $f(1) = 0.$
Taking a $\sigma$-finite measure $\mu$ dominating $P$ and $Q$,
we write 
Radon-Nikod\'ym derivatives as 
$p (x): = d P / d \mu (x)$ and $q (x) := d Q /d \mu  (x) . $
An $f$-divergence~\cite{Csiszar67,1705001} between $P$ and $Q$ is defined by
\[
	D_f (P , Q) 
	:=
	\int_\Omega
	f
	\left(  
	\frac{  p (x) }{q (x) } 
	\right)
	q(x)
	d \mu (x)
	,
\]
where the integral on the RHS is independent of the choice of $\mu$
and the following conventions are adopted:
\begin{gather*}
	f^\ast (0)
	:=
	\lim_{t \to \infty} \frac{ f(t) }{t} ,
	\\
	 0 \cdot 
	 f 
	 \left(    
	 \frac{ p  }{ 0  }
	 \right)
	 = 
	 p f^\ast (0) ,
	 \quad
	 0 \cdot f^\ast (0) = 0.
\end{gather*}
For later use, we fix an $f$ such that $D_f (P , Q) < \infty$ 
for each $P$ and $Q .$
An example of such an $f$ is given by $f(t) = (\sqrt{t} -1)^2,$ and 
the corresponding $f$-divergence is the Hellinger distance
\[
	H (P,Q)
	=
	\int_\Omega
	\left(
	\sqrt{ p(x) }  - \sqrt{ q(x) }
	\right)^2
	d\mu (x)
	\leq 2 .
\]

The following theorem gives necessary and sufficient conditions for the sufficiency of a statistic.

\begin{theo}
\label{theo:sufstat}
Let $\mathcal{P} = \{ P_\theta \}_{\theta \in \Theta}$ be a statistical model
on a measurable space $(\Omega , \mathscr{B})$
dominated by a pivotal probability measure $\lambda$,
and let $T \colon \Omega \to \Omega_T$
be a $\mathscr{B} / \mathscr{B}_T$-measurable statistic
and $P^T_\theta (\cdot) := P_\theta (T^{-1} (\cdot)).$
Then the following conditions are equivalent:
\begin{enumerate}[(i)]
\item
$T$ is sufficient;
\item
for every $P_{\theta} \in \mathcal{P}$,
there exists a $\mathscr{B}_T$-measurable real valued function $g_\theta( \cdot)$
such that
\begin{align*}
	\frac{ dP_\theta }{ d \lambda  } 
	(x)
	=
	g_\theta (T(x))
	\quad
	\lambda \text{-}\mathrm{a.e.;}
\end{align*}
\item
for every $P_{\theta} \in \mathcal{P}$,
\begin{align*}
	\frac{ dP_\theta }{ d \lambda  } (x)
	=
	\frac{  d P_\theta^T }{ d \lambda^T }
	(T(x)) 
	\quad
	\lambda\text{-}\mathrm{a.e.},
\end{align*}
where $\lambda^T(\cdot) := \lambda(T^{-1} (\cdot )) ;$
\item
$
D_f ( P_{\theta_1} ,  P_{\theta_2}  ) 
=
D_f ( P_{\theta_1}^T ,  P_{\theta_2}^T  ) 
$
$(\forall P_{\theta_1},  \forall P_{\theta_2} \in \mathcal{P}).$
\end{enumerate}
\end{theo}
\begin{proof}
The equivalence (i)$\Leftrightarrow$(ii)$\Leftrightarrow$(iv) is 
well-known~\cite{halmos1949application,kullbackleibler1951,bahadur1954,1705001}.
The implication (iii)$\Rightarrow$(ii) is obvious.
Let us show (ii)$\Rightarrow$(iii).
For each $F \in \mathscr{B}_T$, we have
\begin{align*}
	P_\theta^T (F)
	&=
	\int_{\Omega}
	\chi_F(T(x))
	\frac{ dP_\theta }{ d \lambda  } (x)
	d \lambda(x)
	\\
	&=
	\int_{\Omega}
	\chi_F(T(x))
	g_\theta (T(x))
	d \lambda(x)
	\\
	&=
	\int_{\Omega_T}
	\chi_F(t)
	g_\theta (t)
	d \lambda^T(t),
\end{align*}
where $\chi_F (\cdot)$ is the indicator function for $F.$
This implies that $g_\theta (t) = d P_\theta^T /d \lambda^T (t)$ $\lambda^T$-a.e.
and we obtain (iii).
\end{proof}

Let $\mathcal{P}$ be a statistical model on $(\Omega , \mathscr{B})$.
A statistic 
$T \in \mathbb{M}(  (\Omega , \mathscr{B}) \to (\Omega_T , \mathscr{B}_T) )$
is said to be \textit{minimal sufficient} if 
$T$ is sufficient and for each sufficient statistic 
$S \in   \mathbb{M}(  (\Omega , \mathscr{B}) \to (\Omega_S , \mathscr{B}_S) )$
there exists a map
$f \in \mathbb{M}(  (\Omega_S , \mathscr{B}_S) \to (\Omega_T , \mathscr{B}_T) )$
such that $T(x) = f   (S(x))$
$\mathcal{P}$-a.e.
A minimal sufficient statistic can be interpreted to capture the information about 
the statistical model $\mathcal{P}$ in the least redundant manner.

The following theorem gives a sufficient condition for the existence of a
minimal sufficient statistic.

\begin{theo}[Lehmann and Scheff\'{e}~\cite{lehmann1950}, Bahadur~\cite{bahadur1954}]
\label{theo:lsb}
Let $\mathcal{P} = \{ P_\theta  \}_{\theta \in \Theta}  $ be a statistical model 
on an outcome space $(\Omega , \mathscr{B})$.
Suppose that there exists a countable subset
$\{ P_{\theta_i}  \}_{i \geq 1} \subset \mathcal{P}$
dense in $\mathcal{P}$ with respect to the following metric:
\begin{equation}
	d (P , Q)
	:=
	\sup_{E \in \mathscr{B} } 
	|
	P(E) - Q(E)
	| .
	\label{eq:ddef}
\end{equation}
Then a statistic 
$T \in \mathbb{M} (  (\Omega , \mathscr{B}) \to  (\mathbb{R}^\infty , \mathscr{B} (\mathbb{R}^\infty) )  )$
defined by
\begin{equation}
	T(x) 
	:=
	\left(
	\frac{ d P_{\theta_i} }{ d\lambda }
	(x)
	\right)_{i \geq 1}
	\in
	\mathbb{R}^\infty
	\notag 
\end{equation}
is a minimal sufficient statistic.
Here $(\mathbb{R}^\infty , \mathscr{B} (\mathbb{R}^\infty) )$ is the countable product space of the real line
$(\mathbb{R} , \mathscr{B} (\mathbb{R})) $
and $\lambda := \sum_{i \geq 1} c_i P_{\theta_i}$ with $c_i >0$ and $\sum_{i \geq 1} c_i =1$.
\end{theo}

\subsection{Positive-operator valued measure}
Throughout this paper, we fix a \textit{separable} (i.e. $\dim \mathcal{H} \leq \aleph_0$)
Hilbert space $\mathcal{H}$ and denote the set of bounded operators on 
$\mathcal{H}$
by
$\mathcal{L}(\mathcal{H}).$
A positive trace class operator $\rho$ with unit trace is called a \textit{state} on $\mathcal{H}$,
and the set of states on $\mathcal{H}$ is denoted by $\mathcal{S} (\mathcal{H}).$

Let $(\Omega, \mathscr{B})$ be a measurable space.
A positive-operator valued measure (POVM) $A$ is a mapping
$A \colon \mathscr{B} \to \mathcal{L} (\mathcal{H})$
such that
\begin{enumerate}
\item[(i)]
$A(E) \geq O$ for each $E \in \mathscr{B} ; $
\item[(ii)]
$A(\Omega) = I$;
\item[(iii)]
for each countable disjoint family $\{ E_i \} \subset \mathscr{B}$,
$A(\bigcup_i E_i  ) = \sum_{i} A( E_i ) ,$
where the RHS is convergent in the sense of the weak operator topology. 
\end{enumerate}
Here, $O$ and $I$ are the zero and identity operators, respectively.
The triple $(\Omega, \mathscr{B} , A)$ is also called a POVM.



Let $(\Omega , \mathscr{B} , A)$ be a POVM.
For each $\rho \in \mathcal{S} (\mathcal{H} )$
we define a probability measure $P^A_{\rho} (\cdot)$
with the outcome space $(\Omega , \mathscr{B})$ by
$P^A_{\rho} (E) := \tr [ \rho A(E) ]$
for each $E \in \mathscr{B} .$
Then $A$ induces a natural statistical model 
$\mathcal{P}^A := \{ P^A_\rho  \}_{\rho \in \mathcal{S} (\mathcal{H})}$,
which is the set of possible outcome distributions when we perform the measurement $A$.

From the separability of $\mathcal{H}$,
there exists a sequence $\{ \rho_i \}_{i\geq 1}$ in $\stsp$ 
dense with respect to the trace norm topology.
Taking arbitrary $\{ c_i \}_{i\geq 1}$ such that
$
	c_i>0
$
and
$
	\sum_{i \geq 1} c_i = 1,
$
e.g. $c_i = 2^{-i},$
we define a state $\rho_\ast := \sum_{i \geq 1} c_i \rho_i .$
Throughout this paper, we fix such $\{  \rho_i \}_{i \geq 1}$ 
and $\rho_\ast$.

From the definition of $\rho_\ast ,$ 
$P^A_{\rho_\ast} = \sum_{i\geq 1} c_i P^A_{\rho_i}$
and
the following proposition immediately follows.
\begin{prop}
\label{prop:aconti}
Let $(\Omega , \mathscr{B} , A)$ be a POVM.
Then $\mathcal{P}^A \ll P^A_{\rho_\ast}$,
i.e. 
$P^A_{\rho_\ast}$ is a pivotal measure for $\mathcal{P}^A .$
\end{prop}

Due to the above proposition, the notions of $A$-a.e., $\mathcal{P}^A$-a.e., and $P^A_{\rho_\ast}$-a.e.
coincide,
and thus we will use them interchangeably.

\subsection{Fuzzy preorder and equivalence relations among POVMs}
Let $(\Omega_1 , \mathscr{B}_1)$ and $(\Omega_2 , \mathscr{B}_2)$
be measurable spaces.
A mapping 
$\kappa(\cdot | \cdot) \colon 
\mathscr{B}_1 \times \Omega_2 
\to
[0,1]
$
is called a \textit{regular Markov kernel} if
\begin{enumerate}
\item[(i)]
$\kappa (E | \cdot)$ is $\mathscr{B}_2$-measurable for each $E \in \mathscr{B}_1$;
\item[(ii)]
$\kappa (\cdot | y )$ is a probability measure for each $y \in \Omega_2 .$ 
\end{enumerate}

Let $(\Omega_1 , \mathscr{B}_1)$ be a measurable space
and let $(\Omega_2, \mathscr{B}_2 , A_2)$ be a POVM.
An $A_2$-\textit{weak Markov kernel} is a mapping 
$\kappa(\cdot | \cdot) \colon 
\mathscr{B}_1 \times \Omega_2 
\to
\mathbb{R}
$
such that
\begin{enumerate}
\item[(i)]
$\kappa (E | \cdot)$ is $\mathscr{B}_2$-measurable for each $E \in \mathscr{B}_1$;
\item[(ii)]
$ 0 \leq   \kappa (E | y)  \leq 1 $ $A_2$-a.e.
for each $E \in \mathscr{B}_1$;
\item[(iii)]
$\kappa (\Omega_1 | y) =1$ $A_2$-a.e.
and
$\kappa (\emptyset | y) =0$ $A_2$-a.e.;
\item[(iv)]
$\kappa (\cup_i E_i | y) = \sum_i \kappa (E_i | y)$
$A_2$-a.e.
for each countable and disjoint $\{  E_i \} \subset \mathscr{B}_1.$
\end{enumerate}

By using the concepts of the regular and weak Markov kernels, 
we introduce fuzzy preorder and equivalence relations as follows~\cite{Dorofeev1997349,Heinonen200577,Jencova2008}.

\begin{defi}
\label{defi:relations}
Let $(\Omega_A , \mathscr{B}_A , A)$ and $(\Omega_B , \mathscr{B}_B , B)$ be 
POVMs.
\begin{enumerate}[(i)]
\item
If there exists a regular Markov kernel 
$
\kappa (\cdot| \cdot) 
\colon
\mathscr{B}_A \times \Omega_B
\to
[0,1]
$
such that
\begin{equation}
	A(E)
	=
	\int_{\Omega_B}
	\kappa (E|y)
	d B(y)
	\quad
	(E \in \mathscr{B}_A),
	\label{eq:fuzzy}
\end{equation}
then we say that $A$ is regularly fuzzier than $B$, denoted by $A \preceq_r B.$
\item
If there exists a $B$-weak Markov kernel 
$
\kappa (\cdot| \cdot) 
\colon
\mathscr{B}_A \times \Omega_B
\to
\mathbb{R}
$
such that the condition~\eqref{eq:fuzzy} holds,
then we say that
$A$ is weakly fuzzier than $B$, denoted by $A \preceq_w B .$
\item
If $A \preceq_r B$ and $B \preceq_r A$, then $A$ and $B$ are said to be regularly fuzzy equivalent,
denoted by $A \simeq_r B.$
\item
If $A \preceq_w B$ and $B \preceq_w A$, then $A$ and $B$ are said to be weakly fuzzy equivalent,
denoted by $A \simeq_w B.$
\end{enumerate}
\end{defi}
Intuitively the relations $A\preceq_r B$ and $A\preceq_w B$ mean that
the measurement of $A$ can be realized by a classical post-processing of
the measurement of $B$.
Apparently, $A \preceq_r B$ (resp. $A \simeq_r B$) implies $A \preceq_w B$ (resp. $A \simeq_w B$). 
It is known~\cite{Dorofeev1997349,Heinonen200577,jencova2009} 
that the regular and weak relations 
$\simeq_r$ and $\simeq_w$ 
(resp. $\preceq_r$ and $\preceq_w$)
are equivalence relations
(resp. preorder relations).
Apparently, the regular relation $A \simeq_r B$
implies the weak relation $A \simeq_w B ,$
while the converse does not necessarily hold.
See Appendix~\ref{sec:app1} for an explicit example of POVMs
that are weakly fuzzy equivalent but not regularly fuzzy equivalent.

A standard Borel space~\cite{srivastava1998course} is 
a measurable space Borel isomorphic to a complete separable metric space.
We call a POVM with a standard Borel outcome space a standard Borel POVM.
By further assuming the standard Borel properties of POVMs,
the weak relations imply the corresponding regular relations as in the following proposition.

\begin{prop}[Remark~4.1 of Ref.~\onlinecite{Jencova2008}]
\label{prop:wimplyr}
Let $A$ and $B$ be POVMs.
\begin{enumerate}
\item
If $A$ is standard Borel,
$A \preceq_w B \Leftrightarrow  A \preceq_r B .$
\item
If $A$ and $B$ are standard Borel,
$A \simeq_w B \Leftrightarrow  A \simeq_r B .$
\end{enumerate}
\end{prop}

In the formulation of the minimal sufficient POVM, 
we will use the regular and weak fuzzy equivalence relations.


\subsection{Sufficient statistics for POVM}
In this subsection, we consider the sufficiency condition of a statistic
for a POVM.
We also show Lemma~\ref{lemm:joint} 
which will be used in Sec.~\ref{sec:minimal}.

Let $(\Omega_A , \mathscr{B}_A ,A)$ be a POVM
and let $T \in \mathbb{M} ((\Omega_A , \mathscr{B}_A ) \to (\Omega_T , \mathscr{B}_T))$
be a statistic.
We define a POVM $A_T$ with the outcome space $(\Omega_T, \mathscr{B}_T)$
by $A_T (\cdot) := A (T^{-1} (\cdot)).$
Since $A_T(\cdot)$ can be written as
\[
	A_T(F)
	=
	\int_{\Omega_A} \chi_F (T(x)) 
	d A(x)
	\quad (E \in \mathscr{B}_T),
\]
we have $A_T \preceq_r A.$
$T$ is said to be a sufficient statistic for $A$ if $T$ is sufficient for 
the statistical model $\mathcal{P}^A = \{  P^A_\rho \}_{\rho \in \stsp}$.

The following lemma, the monotonicity of the $f$-divergence,
states that the $f$-divergence is monotonically decreasing by the classical post-processing.

\begin{lemm}[Theorem~7.1 of Ref.~\onlinecite{Jencova2008}]
\label{lemm:Df}
Let $A$ and $B$ be POVMs.
\begin{enumerate}[(i)]
\item
$A \preceq_w  B$ implies
$D_f (P^A_\rho , P^A_\sigma ) \leq D_f (P^B_\rho , P^B_\sigma )$
$(\forall \rho , \forall \sigma \in \stsp ).$
\item
$A \simeq_w  B$ implies
$D_f (P^A_\rho , P^A_\sigma ) = D_f (P^B_\rho , P^B_\sigma )$
$(\forall \rho  , \forall \sigma \in \stsp ).$
\end{enumerate}
\end{lemm}

The next theorem characterizes the sufficiency of a statistic for a POVM.

\begin{theo}
\label{theo:povmsuff}
Let $(\Omega_A , \mathscr{B}_A ,A)$ be a POVM
and let $T \in \mathbb{M} ((\Omega_A , \mathscr{B}_A ) \to (\Omega_T , \mathscr{B}_T))$
be a statistic.
Then the following conditions are equivalent.
\begin{enumerate}[(i)]
\item
$T$ is sufficient for $A$;
\item
for each $\rho \in \stsp$ there exists a $\mathscr{B}_T$-measurable map $g_\rho (\cdot)$
such that
\[
	\frac{  d P^A_\rho }{d P^A_{\rho_\ast}}
	(x)
	=
	g_\rho ( T(x))
	\quad
	A\text{-}\mathrm{a.e.};
\]
\item
for each $\rho \in \stsp$,
\[
	\frac{  d P^A_\rho }{d P^A_{\rho_\ast}}
	(x)
	=
	\frac{  d P^{A_T}_\rho }{d P^{A_T}_{\rho_\ast}}
	(T(x))
	\quad
	A\text{-}\mathrm{a.e.};
\]
\item
for each $\rho , \sigma \in \stsp$,
$D_f (P^A_\rho , P^A_{\sigma})  = D_f(P^{A_T}_\rho , P^{A_T}_{\sigma}) ;$
\item
$A \simeq_w A_T .$
\end{enumerate}
Furthermore, if $A$ is a standard Borel POVM the above conditions are equivalent to
\begin{enumerate}
\item[(vi)]
$A \simeq_r A_T .$
\end{enumerate}
\end{theo}

\begin{proof}
From Theorem~\ref{theo:sufstat} and Lemma~\ref{prop:aconti},
the equivalence 
(i)$\Leftrightarrow$(ii)$\Leftrightarrow$(iii)$\Leftrightarrow$(iv)
immediately follows.

Let us show (i)$\Rightarrow$(v).
$A_T \preceq_r A  $ is evident from the definition of $A_T.$
Since $T$ is sufficient, 
there exists a conditional probability $P(E|\cdot)$,
which is $\mathscr{B}_T$-measurable for each $E \in \mathscr{B}_A$,
such that
\[
	P^{A}_{\rho} (E)
	=
	\int_{\Omega_T}
	P (E|t)
	\,
	d P^{A_T}_\rho (t) . 
\]
Since $P (\cdot | \cdot )$ satisfies the conditions for the $A_T$-weak Markov kernel
from the definition of the conditional probability,
we obtain $A \preceq_w A_T ,$
and thus (v) holds.

(v)$\Rightarrow$(iv) follows from Lemma~\ref{lemm:Df}.

Let us assume $A$ is a standard Borel POVM.
If we assume (v),
from Proposition~\ref{prop:wimplyr}
we have $A \preceq_r A_T .$
Since $A_T \preceq_r A$ by definition,
we have proved (v)$\Rightarrow$(vi).
The converse (vi)$\Rightarrow$(v) is evident.
\end{proof}

The following lemma assures the existence of a 
POVM corresponding to the joint distribution 
for a given POVM and a regular Markov kernel.


\begin{lemm}
\label{lemm:joint}
Let $(\Omega_1 ,\mathscr{B}_1)$ be a measurable space
and let $(\Omega_2 ,\mathscr{B}_2 , A_2 )$ be a POVM.
Let
$
\kappa (\cdot | \cdot) 
\colon 
\mathscr{B}_1 \times \Omega_2
\to
[0,1]
$
be a regular Markov kernel.
Then the following assertions hold.
\begin{enumerate}
\item
There exists a unique POVM 
$A_{12}$ 
with the product outcome space
$(\Omega_1 \times \Omega_2 , \mathscr{B}_1 \times \mathscr{B}_2)$
such that
\begin{equation}
	A_{12} (E_1 \times E_2 )
	=
	\int_{E_2} 
	\kappa (E_1 | y) 
	d A_2 (y)
	\quad
	(E_1 \in \mathscr{B}_1, \, E_2 \in \mathscr{B}_2).
	\label{eq:kernelprod}
\end{equation}
\item
The canonical projection
\[
\pi_2 \colon 
\Omega_1 \times \Omega_2 
\ni (x, y) 
\mapsto
y \in \Omega_2
\]
is a sufficient statistic for $A_{12}.$
\item
$A_{2} = (A_{12})_{\pi_2} \simeq_r A_{12}  .$
\item
For each $\rho \in \stsp ,$
\[
	\frac{ 
	d P^{A_{12}}_{\rho}   
	}{
	d P^{A_{12}}_{\rho_\ast}
	}
	(x,y)
	=
	\frac{ 
	d P^{A_{2}}_{\rho}   
	}{
	d P^{A_{2}}_{\rho_\ast}
	}
	(y)
	\quad
	A_{12}\text{-}\mathrm{a.e.}
\]
\end{enumerate}
\end{lemm}


\begin{proof}
\begin{enumerate}

\item
Let us consider the following mapping
\begin{gather}
	\tilde{\kappa}
	(F| y)
	:=
	\kappa (F \rvert_{y} |y )
	,
	\quad
	(F \in \mathscr{B}_1 \times \mathscr{B}_2 , y \in \Omega_2) ,
	\notag 
	\\
	F \rvert_{y}
	:=
	\Set{
	x \in \Omega_1 
	|
	(x, y) \in F
	}.
	\notag
\end{gather}
Then $\tilde{\kappa} ( \cdot | y )$ is a probability measure for each $y \in \Omega_2.$
To show the measurability of $\tilde{\kappa} (F| \cdot )$ for each $F \in \mathscr{B}_1 \times \mathscr{B}_2$,
let us define a class 
\[
	\mathscr{D}_{\tilde{\kappa}}
	:=
	\Set{
	F \in   \mathscr{B}_1 \times \mathscr{B}_2
	|
	\text{
	$\tilde{\kappa} (F | \cdot )$
	is $\mathscr{B}_2$-measurable
	}
	}.
\]
Then 
$\mathscr{D}_{\tilde{\kappa}}$
is a Dynkin class,
i.e.
$\mathscr{D}_{\tilde{\kappa}}$
contains $\Omega_1 \times \Omega_2$ and 
is closed under proper differences and countable disjoint unions.
Since 
\begin{equation}
\tilde{\kappa} (E_1 \times E_2 | y) = \kappa (E_1 |y) \chi_{E_2 } (y)
\quad
(E_1 \in \mathscr{B}_1, E_2 \in \mathscr{B}_2),
\label{eq:toch1}
\end{equation}
$\mathscr{D}_{\tilde{\kappa}}$
contains the family of cylinder sets
$\mathscr{C} := \Set{ E_1 \times E_2 |  E_1 \in \mathscr{B}_1, E_2 \in \mathscr{B}_2 }$,
which is a multiplicative class.
Then the Dynkin class theorem assures that 
$\mathscr{D}_{\tilde{\kappa}} = \mathscr{B}_1 \times \mathscr{B}_2$,
i.e. $\tilde{\kappa} (\cdot | \cdot)$ is a regular Markov kernel.
Let us define a POVM $A_{12}$ by
\begin{equation}
	A_{12} (F)
	:=
	\int_{\Omega_2}
	\tilde{\kappa}
	(F|y)
	d A_2 (y),
	\quad
	(F \in \mathscr{B}_1 \times \mathscr{B}_2 ) .
	\label{eq:kappatilde}
\end{equation}
Then from Eq.~\eqref{eq:toch1}, 
$A_{12}$ satisfies the desired condition~\eqref{eq:kernelprod}.
The uniqueness of $A_{12}$ immediately follows from the Dynkin class theorem.
\item
Let us denote the conditional probability of $P^{A_{12}}_{\rho}$
for given $\pi_2  = y$ by
$P^{A_{12}}_{\rho} (F|y)$ $(F \in \mathscr{B}_1 \times \mathscr{B}_2)$.
Then for each $E_1 \in \mathscr{B}_1$ and $E_2 , E_2^\prime \in \mathscr{B}_2$,
we have
\begin{align*}
	P_{\rho}^{A_{12}}
	(E_1 \times (E_2 \cap E_2^\prime))
	&=
	\int_{E_2^\prime}
	\kappa (E_1|y) 
	\chi_{E_2} (y) 
	d P_{\rho}^{A_{2}} (y)
	\\
	&=
	\int_{E_2^\prime}
	P^{ A_{12} }_{\rho}
	(E_1 \times E_2|y)
	d P_{\rho}^{A_{2}} (y).
\end{align*}
Thus we have
\begin{gather*}
	P^{ A_{12} }_{\rho}
	(E_1 \times E_2|y)
	=
	\kappa (E_1|y) 
	\chi_{E_2} (y) 
	\quad
	P^{ A_{2} }_{\rho}\text{-a.e.}
\end{gather*}
and
\begin{gather*}
	P^{ A_{12} }_{\rho_\ast}
	(E_1 \times E_2|y)
	=
	\kappa (E_1|y) 
	\chi_{E_2} (y) 
	\quad
	A_2\text{-a.e.}
\end{gather*}
Therefore we obtain
\[
	P^{ A_{12} }_{\rho}
	(E_1 \times E_2|y)
	=
	P^{ A_{12} }_{\rho_\ast}
	(E_1 \times E_2|y)
	\quad
	P^{ A_{2} }_{\rho}\text{-a.e.}
\] 
To show that
$
	P^{ A_{12} }_{\rho}
	(F|y)
$
can be taken independent of $\rho$,
we define a class
\[
	\mathscr{D}_{\rho}
	:=
	\Set{
	F \in \mathscr{B}_1 \times \mathscr{B}_2
	|
	P^{ A_{12} }_{\rho}
	(F|y)
	=
	P^{ A_{12} }_{\rho_\ast}
	(F|y)
	\quad
	P^{ A_{2} }_{\rho}\text{-a.e.}
	} .
\]
Then
$\mathscr{D}_{\rho}$
is a Dynkin class that contains the family of cylinder sets $\mathscr{C}.$
Therefore the Dynkin class theorem assures that
$
	P^{ A_{12} }_{\rho}
	(F|y)
	=
	P^{ A_{12} }_{\rho_\ast}
	(F|y)
$
$
	P^{ A_{2} }_{\rho}\text{-a.e.}
$
for each $F \in \mathscr{B}_1 \times \mathscr{B}_2$.
This implies that $\pi_2$ is a sufficient statistic.

The assertion 3 is evident from 
Eq.~\eqref{eq:kappatilde} and $A_2 = (A_{12})_{\pi_2} \preceq_r A_{12} .$

The assertion 4 immediately follows from 
the assertion 3
and Theorem~\ref{theo:povmsuff}.\qedhere
\end{enumerate}
\end{proof}

\section{Minimal sufficient POVM}
\label{sec:minimal}
In this section, we define a minimal sufficient POVM
and show the existence and uniqueness of a minimal sufficient POVM
equivalent to a given POVM.
\begin{defi}
\label{defi:minimalpovm}
Let $(\Omega_A , \mathscr{B}_A , A)$ be a POVM.
\begin{enumerate}
\item
$(\Omega_A , \mathscr{B}_A , A)$ is said to be $\simeq_r$-minimal sufficient
if for any POVM 
$(\Omega_B , \mathscr{B}_B , B)$ regularly fuzzy equivalent to $A$
there exists a measurable map 
$f \in \mathbb{M} ( (\Omega_B , \mathscr{B}_B) \to (\Omega_A , \mathscr{B}_A)  )$
such that $B_f = A .$
\item
$(\Omega_A , \mathscr{B}_A , A)$ is said to be $\simeq_w$-minimal sufficient
if for any POVM 
$(\Omega_B , \mathscr{B}_B , B)$ weakly fuzzy equivalent to $A$
there exists a measurable map 
$f \in \mathbb{M} ( (\Omega_B , \mathscr{B}_B) \to (\Omega_A , \mathscr{B}_A)  )$
such that $B_f = A .$
\end{enumerate}
\end{defi}
Since $\simeq_w$ is a relation less restrictive than $\simeq_r ,$
any $\simeq_w$-minimal sufficient POVM is a $\simeq_r$-minimal sufficient POVM
by definition.

We remark that the measurable map $f$ in Definition~\ref{defi:minimalpovm}
is a sufficient statistic for the statistical model $\mathcal{P}^B $
due to
Theorem~\ref{theo:povmsuff} and $B_f = A \simeq_w B$.


The minimal sufficiency of a POVM can be interpreted as a
generalization of the minimal sufficiency of a statistic
in the sense that we consider a more general class of classical post-processings
which includes taking statistics.
Intuitively, a minimal sufficient POVM $A$ is the least redundant POVM
among POVMs fuzzy equivalent to $A$.


\subsection{Lehmann-Scheff\'{e}-Bahadur statistic}
Let $(\Omega_A , \mathscr{B}_A , A)$ be a POVM.
Corresponding to Theorem~\ref{theo:lsb} we define a statistic
$T\in \mathbb{M} ( (\Omega_A , \mathscr{B}_A) \to (\mathbb{R}^\infty , \mathscr{B} (\mathbb{R}^\infty)) )$
by
\begin{equation}
	T (x)
	:=
	\left(
	\frac{
	d P^A_{\rho_i} 
	}{
	d P^A_{\rho_\ast}
	}
	(x)
	\right)_{i \geq 1}
	\in \mathbb{R}^\infty .
	\label{eq:povmlsb}
\end{equation}
By taking authors\rq{} names of 
Refs.~\onlinecite{lehmann1950,bahadur1954},
we call the statistic~\eqref{eq:povmlsb} 
a \textit{Lehmann-Scheff\'{e}-Bahadur (LSB) statistic}
for $A$.

The following theorem is the first main result of this paper.

\begin{theo}
\label{theo:minimal_exists}
Let $(\Omega_A , \mathscr{B}_A , A)$ be an arbitrary POVM
and let 
$T$ be
the LSB statistic given by~\eqref{eq:povmlsb}.
Then the following assertions hold.
\begin{enumerate}
\item
$T$ is a minimal sufficient statistic for $\mathcal{P}^A$.
\item
$A_T$ is a $\simeq_w$-minimal sufficient POVM,
and therefore a $\simeq_r$-minimal sufficient POVM.
\end{enumerate}
\end{theo}
\begin{proof}
\begin{enumerate}
\item
From Theorem~\ref{theo:lsb}, it is sufficient to show 
that $\{ P^A_{\rho_i}  \}_{i \geq 1}$ is dense in $\mathcal{P}^A$
with respect to the metric $d(\cdot, \cdot)$ given by~\eqref{eq:ddef}.
For each $\rho , \sigma \in \stsp$ and each $E \in \mathscr{B}_A$ we have
\begin{align*}
	 | P^A_\rho (E) - P^A_\sigma (E) |
	&=
	| \tr [ (\rho - \sigma ) A(E) ] |
	\\
	&\leq
	|| \rho - \sigma ||_{\tr} ||A(E)||
	\leq
	|| \rho - \sigma ||_{\tr} ,
\end{align*}
where 
$||\cdot ||$ and
$||\cdot ||_{\tr}$ 
are the operator and trace norms, respectively.
Thus we obtain
\begin{equation}
	d (P^A_\rho , P^A_\sigma) \leq || \rho - \sigma ||_{\tr} 
	.
	\label{eq:dleq}
\end{equation}
Since $\{  \rho_i \}_{i \geq 1}$ is dense in $\stsp$ with respect to the trace norm,
the inequality~\eqref{eq:dleq} implies that
$\{ P^A_{\rho_i}  \}_{i \geq 1}$ is dense with respect to $d(\cdot , \cdot)$
in $\mathcal{P}^A$,
and we have shown the minimal sufficiency of $T$.

\item
From Theorem~\ref{theo:povmsuff} and the sufficiency of $T$, we have 
\[
	\frac{  
	d P^A_{\rho_i} 
	}{
	d P^A_{\rho_\ast}
	}
	(x)
	=
	\frac{  
	d P^{A_T}_{\rho_i} 
	}{
	d P^{A_T}_{\rho_\ast}
	}
	(T(x)) 
	\quad
	A\text{-a.e.}
\]
for each $i \geq 1.$ 
From the definition of the LSB statistic~\eqref{eq:povmlsb},
this implies 
\[
	T(x)
	=
	\left(
	\frac{  
	d P^{A_T}_{\rho_i} 
	}{
	d P^{A_T}_{\rho_\ast}
	}
	(T(x))
	\right)_{i \geq 1} 
	\quad
	A\text{-a.e.,}
\]
or
\begin{equation}
	t
	=
	\left(
	\frac{  
	d P^{A_T}_{\rho_i} 
	}{
	d P^{A_T}_{\rho_\ast}
	}
	(t)
	\right)_{i \geq 1} 
	\quad
	A_T\text{-a.e.}
	\label{eq:tself}
\end{equation}

To show the $\simeq_w$-minimal sufficiency of $A_T$,
let $(\Omega_B , \mathscr{B}_B , B)$ be an arbitrary POVM 
weakly fuzzy equivalent to $A_T.$
Since the outcome space $(\mathbb{R}^\infty, \mathscr{B}(\mathbb{R}^\infty))$
for $A_T$ is a standard Borel space,
there exists a regular Markov kernel 
$
\kappa (\cdot | \cdot) 
\colon 
\mathscr{B} (\mathbb{R}^\infty) \times \Omega_B
\to
[0,1]
$
such that
\begin{equation}
	A_T (F)
	=
	\int_{\Omega_B}
	\kappa(F|y)
	d B (y)
	\label{eq:Bkernel} 
\end{equation}
for each $F \in \mathscr{B} (\mathbb{R}^\infty) .$
We define a POVM 
$
(
 \mathbb{R}^\infty \times \Omega_B  ,
\mathscr{B} ( \mathbb{R}^\infty) \times \mathscr{B}_B  ,
C
)$
by
\[
	C(E)
	:=
	\int_{\Omega_B}
	\kappa(E \rvert_{y} |y)
	d B (y)
	\quad 
	(E \in \mathscr{B} ( \mathbb{R}^\infty) \times \mathscr{B}_B) .
\]
Then from Lemma~\ref{lemm:joint}, we have $C \simeq_r B = C_{\pi_{B} }$
and
\begin{equation}
	\frac{
	d P^C_{\rho_i}
	}{
	d P^C_{\rho_\ast}
	}
	(t,y)
	=
		\frac{  
	d P^{B }_{\rho_i} 
	}{
	d P^{ B }_{\rho_\ast}
	}
	(y)
	\quad
	C\text{-a.e.,}
	\label{eq:Cdiff}
\end{equation}
for each $i \geq 1 ,$
where $\pi_{ B } \colon   \mathbb{R}^\infty \times   \Omega_B  \to   \Omega_B$
is the canonical projection.
On the other hand, from Eq.~\eqref{eq:Bkernel} we have 
$C_{\pi_{A_T}} = A_T \simeq_w B \simeq_r C $,
where 
$\pi_{A_T} \colon \mathbb{R}^\infty  \times \Omega_B   \to \mathbb{R}^\infty$
is the canonical projection,
and Theorem~\ref{theo:povmsuff} assures that 
\begin{equation}
	\frac{
	d P^C_{\rho_i}
	}{
	d P^C_{\rho_\ast}
	}
	(t,y)
	=
	\frac{  
	d P^{A_T}_{\rho_i} 
	}{
	d P^{A_T}_{\rho_\ast}
	}
	(t)
	\quad
	C\text{-a.e.}
	\label{eq:Cdiff2}
\end{equation}
for each $i \geq 1.$
From Eqs.~\eqref{eq:tself},~\eqref{eq:Cdiff} and \eqref{eq:Cdiff2},
we obtain
\begin{equation}
	t
	=
	\left(
		\frac{  
	d P^{B}_{\rho_i} 
	}{
	d P^{B}_{\rho_\ast}
	}
	(y)
	\right)_{i \geq 1}
	= :
	f (y)
	\quad 
	C\text{-a.e.,}
	\label{eq:teqf}
\end{equation}
where 
$f \in \mathbb{M} ((\Omega_B , \mathscr{B}_B) 
\to 
(\mathbb{R}^\infty ,
\mathscr{B}  ( \mathbb{R}^\infty  )
)
).$
Then for each $E \in  \mathscr{B}  ( \mathbb{R}^\infty  ) $, we have
\begin{align}
	B_f (E)
	&=
	B(f^{-1} (E) )
	\notag 
	\\
	&=
	C (  \mathbb{R}^\infty \times  f^{-1} (E)  )
	\notag
	\\
	&=
	\int_{   \mathbb{R}^\infty  \times  \Omega_B }
	\chi_E ( f(y) )
	d C (t ,y)
	\notag
	\\
	&=
	\int_{ \mathbb{R}^\infty  \times  \Omega_B}
	\chi_E ( t )
	d C (t,y)
	\label{eq:tubo}
	\\
	&= 
	A_T(E),
	\notag
\end{align}
where we have used Eq.~\eqref{eq:teqf} in deriving the equality~\eqref{eq:tubo}.
Thus we obtain $B_f = A_T,$
and we have shown the $\simeq_w$-minimal sufficiency of $A_T$.
\qedhere
\end{enumerate}
\end{proof}

\subsection{Uniqueness up to almost isomorphism}
In order to formulate the uniqueness of the minimal sufficient POVM,
we introduce a concept of almost isomorphism as follows.

\begin{defi}
\label{defi:isomorphism}
Let $(\Omega_i , \mathscr{B}_i , A_i)$
$(i = 1,2)$ be POVMs.
\begin{enumerate}[(i)]
\item
A $\mathscr{B}_1 / \mathscr{B}_2$-bimeasurable bijection $f \colon \Omega_1 \to \Omega_2$
is called a strict isomorphism if
$(A_1)_f = A_2$.
If there exists such a strict isomorphism,
$(\Omega_1 , \mathscr{B}_1 , A_1)$ and $(\Omega_2 , \mathscr{B}_2 , A_2)$
are said to be strictly isomorphic, denoted by
$(\Omega_1 , \mathscr{B}_1 , A_1) \approx (\Omega_2 , \mathscr{B}_2 , A_2).$
\item
$(\Omega_1 , \mathscr{B}_1 , A_1)$ and $(\Omega_2 , \mathscr{B}_2 , A_2)$
are said to be almost isomorphic, written 
$(\Omega_1 , \mathscr{B}_1 , A_1)  \sim  (\Omega_2 , \mathscr{B}_2 , A_2)$,
if there exist measurable subsets $\tilde{\Omega}_i \in \mathscr{B}_i $
$(i=1,2)$
such that $A_i (\tilde{\Omega}_i) = I$
and the restrictions 
$(
\tilde{\Omega}_i ,
\tilde{\Omega}_i \cap \mathscr{B}_i ,
A_i \rvert_{\tilde{\Omega}_i}
)$
$(i = 1,2)$
are strictly isomorphic.  
Here 
$
\tilde{\Omega}_i \cap \mathscr{B}_i 
:= 
\Set{ \tilde{\Omega}_i \cap E | E \in \mathscr{B}_i  }
$
and 
$A_i \rvert_{\tilde{\Omega}_i}$
is the restriction of $A_i$ to
$\tilde{\Omega}_i \cap \mathscr{B}_i .$
We call a measurable subset $\tilde{\Omega}_i$ with $A_i (\tilde{\Omega}_i) = I$ a
full measure set.
\end{enumerate}
\end{defi}

The above definition is a straightforward generalization of the corresponding concepts
known in classical measures~\cite{ito1984introduction,bogachev2007}.
The relations $\approx$ and $\sim$ are equivalence relations,
which can be proved in a similar manner as for the classical measure
(e.g. Sec.~2.4 of Ref.~\onlinecite{ito1984introduction}).
Intuitively, these concepts correspond to the relabeling of the measurement outcomes.

The following proposition gives the relationship between
these isomorphisms and the fuzzy equivalence relation. 
\begin{prop}
\label{prop:relationorder}
Let $(\Omega_A , \mathscr{B}_A , A)$ and 
$(\Omega_B , \mathscr{B}_B , B)$
be POVMs.
Then the following implications hold.
\begin{align*}
(\Omega_A , \mathscr{B}_A , A)
\approx
(\Omega_B , \mathscr{B}_B , B)
&\Rightarrow
(\Omega_A , \mathscr{B}_A , A)
\sim
(\Omega_B , \mathscr{B}_B , B)
\\
&\Rightarrow
A
\simeq_r
B.
\end{align*}
\end{prop}
\begin{proof}
The first implication is evident from the definitions of the strict and almost isomorphisms.
Let us assume
$(\Omega_A , \mathscr{B}_A , A)
\sim
(\Omega_B , \mathscr{B}_B , B)
$
and show $A \simeq_r B.$
Let $\tilde{\Omega}_A \in \mathscr{B}_A$ and 
$\tilde{\Omega}_B \in \mathscr{B}_B$
be full measure sets such that
$
(\tilde{\Omega}_A , \tilde{\Omega}_A \cap \mathscr{B}_A , A \rvert_{\tilde{\Omega}_A}  )
\approx
(\tilde{\Omega}_B , \tilde{\Omega}_B \cap \mathscr{B}_B , B \rvert_{\tilde{\Omega}_B}  )
.
$
We first show $A \simeq_r A \rvert_{\tilde{\Omega}_A}.$
Since the identity map
\[
	\iota
	\colon
	\tilde{\Omega}_A
	\ni 
	x
	\mapsto
	x
	\in 
	\Omega_A
\]
is $\tilde{\Omega}_A \cap \mathscr{B}_A /\mathscr{B}_A $-measurable,
we have $A = ( A \rvert_{ \tilde{\Omega}_A }  )_{\iota} \preceq_r A \rvert_{ \tilde{\Omega}_A }. $
We fix a point $x_0 \in \tilde{\Omega}_A$ and define a mapping
$\tilde{ \iota } \colon \Omega_A \to \tilde{\Omega}_A$
by
\begin{equation}
	\tilde{\iota}
	(x)
	:=
	\begin{cases}
	x & (x \in \tilde{\Omega}_A) ;
	\\
	x_0 & (x \in \Omega_A \setminus \tilde{\Omega}_A) .
	\end{cases}
	\label{eq:tiota}
\end{equation}
Then for each $E \in \tilde{\Omega}_A \cap \mathscr{B}_A$ we have
\[
	\tilde{\iota}^{-1}
	(E)
	=
	\begin{cases}
	E \cup (\Omega_A \setminus \tilde{\Omega}_A)  & (x_0 \in E) ;
	\\
	E & (x_0 \notin E) .
	\end{cases}
\]
Hence $\tilde{\iota}$ is $\mathscr{B}_A / \tilde{\Omega}_A \cap \mathscr{B}_A$-measurable
and $A \rvert_{ \tilde{\Omega}_A } = A_{\tilde{\iota}} \preceq_r A . $
Thus we have shown $A \simeq_r A \rvert_{ \tilde{\Omega}_A } .$

Similarly we can prove $B \simeq_r B \rvert_{\tilde{\Omega}_B }$.
Since the equivalence $A \rvert_{ \tilde{\Omega}_A } \simeq_r  B \rvert_{\tilde{\Omega}_B }$
immediately follows from the definition of the strict isomorphism,
we obtain 
$
A 
\simeq_r
A \rvert_{ \tilde{\Omega}_A } 
\simeq_r
B \rvert_{\tilde{\Omega}_B }
\simeq_r
B.
$
\end{proof}

We next show that 
the minimal sufficiency condition for a POVM is invariant 
under almost isomorphism.

\begin{prop}
\label{prop:isomini}
Let $(\Omega_A , \mathscr{B}_A , A)$ be a $\simeq_w$-minimal sufficient 
(resp. $\simeq_r$-minimal sufficient)
POVM
and let 
$(\Omega_B , \mathscr{B}_B , B)$
be a POVM almost isomorphic to
$(\Omega_A , \mathscr{B}_A , A).$
Then $B$ is also $\simeq_w$-minimal sufficient
(resp. $\simeq_r$-minimal sufficient).
\end{prop}
\begin{proof}
We prove the assertion for the $\simeq_w$-minimal sufficient POVM;
the proof for $\simeq_r$-minimal sufficiency can be 
obtained by replacing $\simeq_w$ to $\simeq_r$ in the following proof.

Let $\tilde{\Omega}_A \in \mathscr{B}_A$ and 
$\tilde{\Omega}_B \in \mathscr{B}_B$ 
be full measure sets such that 
$
(\tilde{\Omega}_A , \tilde{\Omega}_A \cap \mathscr{B}_A , A \rvert_{\tilde{\Omega}_A}  )
\approx
(\tilde{\Omega}_B , \tilde{\Omega}_B \cap \mathscr{B}_B , B \rvert_{\tilde{\Omega}_B}  )
$
by a strict isomorphism $f \colon \tilde{\Omega}_A \to \tilde{\Omega}_B.$
To show the $\simeq_w$-minimal sufficiency of $B$, 
we take an arbitrary POVM $(\Omega_C , \mathscr{B}_C , C)$
weakly fuzzy equivalent to $B .$
Then from Proposition~\ref{prop:relationorder},
$C$ is also weakly fuzzy equivalent to $A$ and there exists a mapping
$g \in \mathbb{M} ( (\Omega_C , \mathscr{B}_C)  \to (\Omega_A, \mathscr{B}_A)   )$
such that $C_g = A .$
By fixing a point $x_0 \in \tilde{\Omega}_A$ we define a mapping 
$\tilde{\iota} \colon \Omega_A \to \tilde{\Omega}_A$
by Eq.~\eqref{eq:tiota}.
Then $\tilde{\iota}$ is a measurable map such that
$A \rvert_{\tilde{\Omega}_A} = A_{\tilde{\iota}}$.
Let us denote by $\tilde{f}$  
the isomorphism $f$
regarded as a map from $\tilde{\Omega}_A$ to $\Omega_B .$
Then $\tilde{f}$ is $\tilde{\Omega}_A \cap \mathscr{B}_A / \mathscr{B}_B$-measurable
and $(A \rvert_{ \tilde{\Omega}_A })_{\tilde{f}} = B .$
Thus we have
$
	B 
	=
	(A \rvert_{ \tilde{\Omega}_A })_{\tilde{f}}
	=
	(  A_{\tilde{\iota}}    )_{\tilde{f}}  
	=
	\left(  (C_g )_{\tilde{\iota}} \right)_{\tilde{f}}
	=
	C_{  \tilde{f} \circ  \tilde{\iota}  \circ g  } ,
$
and $B$ is $\simeq_w$-minimal sufficient.
\end{proof}

Now we prove the following theorem
which is the second main result of this paper.
\begin{theo}
\label{theo:exunique}
\begin{enumerate}
\item
Let 
$(\Omega_A, \mathscr{B}_A , A )$
an arbitrary POVM.
Then there exists a $\simeq_w$-minimal sufficient POVM
$\tilde{A}$
weakly fuzzy equivalent to $A$.
Furthermore such $\tilde{A}$ is unique
up to almost isomorphism.
\item
Let 
$(\Omega_A, \mathscr{B}_A , A )$
a standard Borel POVM.
Then there exists a $\simeq_r$-minimal sufficient POVM
$\tilde{A}$
regularly fuzzy equivalent to $A$.
Furthermore such $\tilde{A}$ is unique
up to almost isomorphism.
\end{enumerate}
\end{theo}

\begin{proof}
\begin{enumerate}
\item
Let $T\colon \Omega_A \to \mathbb{R}^\infty$
be the LSB statistic given by~\eqref{eq:povmlsb}
and define $\tilde{A} := A_T.$
Then from Theorem~\ref{theo:povmsuff} and Theorem~\ref{theo:minimal_exists},
$\tilde{A}$ is a $\simeq_w$-minimal sufficient POVM weakly fuzzy equivalent to $A$.
From 
the same discussion in Theorem~\ref{theo:minimal_exists},
$\tilde{A}$ satisfies the following condition:
\begin{equation}
	t =
	\left(
	\frac{d P^{\tilde{A}}_{\rho_i} }{ d P^{\tilde{A}}_{\rho_\ast}  }
	(t)
	\right)_{i \geq 1}
	\quad
	\tilde{A}\text{-a.e.}
	\label{eq:self1}
\end{equation}

To prove the uniqueness up to almost isomorphism,
take another $\simeq_w$-minimal sufficient POVM 
$(\Omega_B, \mathscr{B}_B , B )$
weakly fuzzy equivalent to $A$.
Since $B \simeq_w A \simeq_w \tilde{A}$,
there exist mappings
$f \in \mathbb{M } (  (\mathbb{R}^\infty , \mathscr{B} (\mathbb{R}^\infty ) ) \to (\Omega_B , \mathscr{B}_B)  )$
and
$g \in \mathbb{M } (  (\Omega_B , \mathscr{B}_B) \to  (\mathbb{R}^\infty , \mathscr{B} (\mathbb{R}^\infty ) )  )$
such that
$
	\tilde{A}_f = B 
$
and
$
	B_g = \tilde{A}.
$
Then we have $\tilde{A}_{g \circ f} = \tilde{A}$,
and from Theorem~\ref{theo:povmsuff}, $g \circ f \colon \mathbb{R}^\infty \to \mathbb{R}^\infty$
is a sufficient statistic for $\tilde{A} $ and we have
\begin{equation}
	\frac{d  P^{\tilde{A}}_{\rho_i}  }{ d P^{\tilde{A}}_{\rho_\ast}  }
	(t)
	=
	\frac{d  P^{\tilde{A}}_{\rho_i}  }{ d P^{\tilde{A}}_{\rho_\ast}  }
	(g \circ f(t))
	\quad
	\tilde{A}\text{-a.e.}
	\label{eq:self2}
\end{equation}
for each $i \geq 1.$
From Eqs.~\eqref{eq:self1} and \eqref{eq:self2},
we obtain
\[
	t
	=
	\left(
	\frac{d P^{\tilde{A}}_{\rho_i} }{ d P^{\tilde{A}}_{\rho_\ast}  }
	(t)
	\right)_{i \geq 1}
	=
	\left(
	\frac{d P^{\tilde{A}}_{\rho_i} }{ d P^{\tilde{A}}_{\rho_\ast}  }
	(g\circ f (t))
	\right)_{i \geq 1}
	=
	g \circ f (t)
	\quad
	\tilde{A}\text{-a.s.}
\]
Thus there exists a full measure set 
$\tilde{\Omega}_{\tilde{A}} \in \mathscr{B} (\mathbb{R}^\infty)$
such that $( g \circ f ) \rvert_{\tilde{\Omega}_{\tilde{A}}  }$ is an identity map on 
$\tilde{\Omega}_{\tilde{A}}.$
Then $\tilde{\Omega}_B := f (\tilde{\Omega}_{\tilde{A}}) = g^{-1} ( \tilde{\Omega}_{\tilde{A}}  ) \in \mathscr{B}_B $
is a full measure set
and $f \rvert_{\tilde{\Omega}_{\tilde{A}}  }$
is a strict isomorphism between
$(
\tilde{\Omega}_{\tilde{A}} ,
\tilde{\Omega}_{\tilde{A}} \cap \mathscr{B} ( \mathbb{R}^\infty) ,
\tilde{A} \rvert_{\tilde{\Omega}_{\tilde{A}} }
)$ 
and
$(
\tilde{\Omega}_{B} ,
\tilde{\Omega}_{B} \cap \mathscr{B}_B ,
B \rvert_{\tilde{\Omega}_{B} }
)$.
Therefore $\tilde{A}$ and $B$ are almost isomorphic.
\item
As shown in 1, 
the POVM $\tilde{A}$ induced by the LSB statistic for $A$ is  
$\simeq_w$-minimal sufficient POVM,
and therefore $\simeq_r$-minimal sufficient POVM,
weakly fuzzy equivalent to $A$.
Since $A$ and $\tilde{A}$ are standard Borel POVMs,
they are regularly fuzzy equivalent.
The uniqueness up to almost isomorphism can be shown in a similar manner.
\qedhere
\end{enumerate}
\end{proof}


\section{Minimal sufficiency for discrete POVM}
\label{sec:discrete}
In this section, we consider the minimal sufficient condition for
discrete POVMs.

A measurable space $( \Omega , \mathscr{B} )$
is said to be discrete
if $\Omega$ is a countable set
and $\mathscr{B}$ is the power set 
$\mathcal{P}(\Omega) $ 
of $\Omega .$
A POVM $A$ is said to be discrete if
the outcome space of $A$ is a discrete space.
A discrete POVM $(\Omega_A , \mathcal{P}(\Omega_A) , A)$
induces a mapping $A \colon \Omega_A \to \mathcal{L}(\mathcal{H})$
defined by $A(x) := A(\{  x\}) \geq O$
with a completeness condition
\begin{equation}
	\sum_{x \in \Omega_A} A(x) =I.
	\label{eq:completeness}
\end{equation}
On the other hand, a positive-operator valued mapping 
$A\colon \Omega_A \to \mathcal{L}(\mathcal{H})$
with the completeness condition~\eqref{eq:completeness}
induces a discrete POVM by
\[
	A(E)
	:=
	\sum_{x \in E} A(x)
\]
for each $E \in \mathcal{P}( \Omega_A ).$
Thus throughout this section we identify a positive-operator valued mapping 
$A\colon \Omega_A \to \mathcal{L}(\mathcal{H})$
satisfying the completeness condition with a discrete POVM.
We write $p^A_{\rho} (x) := \tr [ \rho A(x) ]$
for each $\rho \in \stsp .$

A discrete POVM 
$A\colon \Omega_A \to \mathcal{L}(\mathcal{H})$
is said to be \textit{non-vanishing}
if $A(x) \neq 0$ for all $x \in \Omega_A .$
Any discrete POVM is almost isomorphic to a non-vanishing POVM.

Since a discrete space is a Standard Borel space, 
the weak relations $\preceq_w$ and $\simeq_w$
coincide with the regular relations
$\preceq_r$ and $\simeq_r ,$
respectively.
Thus for discrete POVMs, fuzzy preorder and equivalence relations are denoted as
$\preceq$ and $\simeq ,$ respectively.

For discrete POVMs 
$A\colon \Omega_A \to \mathcal{L}(\mathcal{H})$
and
$B\colon \Omega_B \to \mathcal{L}(\mathcal{H})$
the relations $\approx$, 
$\sim$,
and  $\preceq$ are simplified as follows.
$A \approx B$ if and only if
there exists a bijection $f \colon \Omega_A \to \Omega_B$
such that 
$A_f = B ,$
where
$A_f (y) = \sum_{x \in \Omega_A} \delta_{y , f(x)}  A (x) .$
$A \sim B$ if and only if 
there exist full-measure subsets 
$\tilde{\Omega}_A \subset \Omega_A$ 
and 
$\tilde{\Omega}_B \subset \Omega_B$
such that
the restrictions 
$A \rvert_{\tilde{\Omega}_A}
$
and
$B \rvert_{ \tilde{\Omega}_B }
$
are non-vanishing
and strictly isomorphic.
$A \preceq B$
if and only if
there exists a matrix $\{ \kappa (x | y) \}_{(x,y) \in \Omega_A \times \Omega_B}$
such that
\begin{gather}
	\kappa (x | y) \geq 0 ,
	\quad
	\sum_{x\in \Omega_A}
	\kappa (x|y)
	=1,
	\label{eq:mm}
	\\
	A(x)
	=
	\sum_{y \in \Omega_B}
	\kappa (x|y) B(y) .
	\notag
\end{gather}
A matrix $\kappa (\cdot | \cdot)$ satisfying the condition~\eqref{eq:mm}
is called a Markov matrix.

The following proposition characterizes the sufficiency condition 
of a statistic for a discrete POVM.
\begin{prop}
\label{prop:discpovm}
Let 
$A\colon \Omega_A \to \mathcal{L}(\mathcal{H})$
be a discrete POVM
and let
$T \colon \Omega_A \to \Omega_T$
be a mapping to a measurable space 
$(\Omega_T , \mathscr{B}_T) .$
We assume that $\mathscr{B}_T$ contains each single point set $\{ t \} $
$(t \in \Omega_T)$.
Then the following conditions are equivalent.
\begin{enumerate}[(i)]
\item
$T$ is sufficient for $A$;
\item
there exist functions
$h (\cdot ) \colon \Omega_A \ni x \mapsto  h(x) \in [0, \infty)$
and 
$
g_{\cdot} (\cdot) \colon
\stsp \times \Omega_T
\ni 
(\rho , t)
\mapsto
g_{\rho} (t)
\in
[0 , \infty)
$
such that
\begin{equation}
	p^A_{\rho} (x)
	=
	h(x)
	g_{\rho} ( T(x) )
	,
	\quad
	(\rho \in \stsp  ,
	x \in \Omega_A);
	\notag
\end{equation}
\item
there exist functions
$h (\cdot ) \colon \Omega_A \ni x \mapsto h(x)  \in [0, \infty)$
and 
$
G  (\cdot) \colon
\Omega_T
\ni t
\mapsto
G  (t)
\in
\mathcal{L}_+ (\mathcal{H})
$
such that
\begin{equation}
	A (x)
	=
	h(x)
	G( T(x) )
	,
	\quad
	(x \in \Omega_A),
	\notag
\end{equation}
where $\mathcal{L}_+ (\mathcal{H}) := \Set{  a \in \mathcal{L} (\mathcal{H}) | a \geq O  };$
\item
$A \simeq A_T .$
\end{enumerate}
\end{prop}
\begin{proof}
The equivalence (i)$\Leftrightarrow$(ii)$\Leftrightarrow$(iv) is evident from Theorem~\ref{theo:povmsuff}.
The implication 
(iii)$\Rightarrow$(ii)
immediately follows by putting $g_{\rho} (t) := \tr [\rho G(t)].$
Let us assume (ii) and show (iii).
If $t \in \Omega_T $ satisfies
that
\begin{equation}
	\exists x \in \Omega_A
	\text{
	such that
	$T(x) =t$ 
	and
	$h(x) \neq 0 ,$
	}
	\label{eq:cond}
\end{equation}
then we have
\begin{equation}
	g_{\rho} (t)
	=
	p^A_\rho (x) /h(x).
	\label{eq:grho}
\end{equation}
Since the RHS of Eq.~\eqref{eq:grho}
is affine and positive with respect to $\rho \in \stsp ,$
according to Refs.~\onlinecite{davies1976quantum,heinosaari2011mathematical},
there exists $G (t) \in \mathcal{L}_+ (\mathcal{H})$ such that
$g_{\rho} (t) = \tr [ \rho G(t)  ]$ for any $\rho \in \stsp .$
For $t \in \Omega_T$ that does not satisfy the condition~\eqref{eq:cond},
we define $G(t)  = O.$
Then we have
\[
	\tr[\rho A(x)]
	=
	p^A_{\rho} (x)
	=
	\tr
	[
	\rho
	h(x) G(T(x))
	]
\]
for each $\rho \in \stsp$ and $x \in \Omega_A ,$
which implies the condition~(iii).
\end{proof}

Corresponding to Theorem~\ref{theo:exunique}
we have the following theorem as to the existence and uniqueness of
a
minimal sufficient POVM.
\begin{theo}
\label{theo:discexuni}
Let
$A\colon \Omega_A \to \mathcal{L}(\mathcal{H})$
be a discrete POVM.
Then there exists a discrete non-vanishing $\simeq_w$-minimal sufficient POVM
$\bar{A}$
fuzzy equivalent to $A$.
Furthermore such $\bar{A}$ is unique up to strict isomorphism.
\end{theo}
\begin{proof}
Since $A$ is almost isomorphic to a non-vanishing POVM,
without loss of generality we can assume $A$ is non-vanishing.
We define an equivalence relation on $\Omega_A$ by
\[
	x \sim_A x^\prime
	:\Leftrightarrow
	\exists c >0,
	A(x) = c A(x^\prime),
\]
and define a mapping 
$S \colon \Omega_A \to  \Omega_A / \sim_A = : \Omega_S $
by
$S (x) := [x] ,$
where $[x]$ is the equivalence class to which $x$ belongs.
Then from the definition of $\sim_A$,
we can write
\[
	A(x) = h (x) G( S (x)) 
\]
for each $x \in \Omega_A$,
where $h(x)>0$ and $O \neq G(s) \in \mathcal{L}_+( \mathcal{H} )$
$(s \in \Omega_S)$,
which implies $S$ is sufficient for $A.$
Therefore if we define
$\bar{A}:= A_S$,
$\bar{A} \simeq A$ and $\bar{A}$ is a non-vanishing discrete POVM.

From Theorem~\ref{theo:exunique} and its prove,
the LSB statistic $T \colon \Omega_S \to \mathbb{R}^\infty$
defined by \eqref{eq:povmlsb} induces a $\simeq_w$-minimal sufficient POVM
$(\mathbb{R}^\infty, \mathscr{B} (\mathbb{R}^\infty) , \bar{A}_T)$
fuzzy equivalent to $\bar{A}.$
Since $T$ is sufficient for $\bar{A}$,
from Proposition~\ref{prop:discpovm},
we can write 
\[
	\bar{A}(s)
	=
	h^\prime (s)
	G^\prime(T(s)) ,
\]
where $h^\prime (s)>0$ and $O \neq G^\prime (t) \in \mathcal{L}_+( \mathcal{H} )$
$(t \in \mathbb{R}^\infty)$.
From the construction of $S$,
this implies that $T$ is injective,
and $T$ is a strict isomorphism between
$(\Omega_S , \mathcal{P}(\Omega_S)  ,  \bar{A})$
and
$(T (\Omega_S)  , T ( \Omega_S ) \cap \mathscr{B} (\mathbb{R}^\infty) , (\bar{A}_S) \rvert_{ T (\Omega_S)  }  ) .$
Note that $T (\Omega_S) \in \mathscr{B} (\mathbb{R}^\infty)$
since $T (\Omega_S)$ is a countable set. 
Thus from
Proposition~\ref{prop:isomini},
$\bar{A}$ is a $\simeq_w$-minimal sufficient POVM.

To show the uniqueness,
let
$B\colon \Omega_B \to \mathcal{L}(\mathcal{H})$
be an arbitrary $\simeq_w$-minimal sufficient and non-vanishing discrete POVM
equivalent to $A$.
Since $\bar{A} \simeq B ,$
there exists a mapping
$f \colon \Omega_S \to \Omega_B$
such that
$A_f = B$. 
If $f$ is not surjective,
there exists an element
$y \in \Omega_B \setminus f (\Omega_S)$
and we have
$B(y) = A_f (y) = \sum_{x \colon f(x) = y} A(x) = O $,
which contradicts the non-vanishing property of $B$.
Thus $f$ is surjective.
Since $f$ is sufficient for $\bar{A}$,
by a similar discussion for $T$,
we can show that $f$ is injective.
Therefore $f$ is a strict isomorphism between 
$\bar{A}$ and $B$,
which completes the proof.
\end{proof}

A discrete POVM
$A\colon \Omega_A \to \mathcal{L}(\mathcal{H})$
is said to be pairwise linearly independent~\cite{martens1990}
if any pair $\{ A(x) , A(x^\prime) \}$
$(x \neq x^\prime)$
is linearly independent.
A pairwise linearly independent POVM is non-vanishing by definition.
The following theorem states that
the minimal sufficiency and the pairwise linearly independence are
almost equivalent.

\begin{theo}
\label{theo:discequi}
Let $A\colon \Omega_A \to \mathcal{L}(\mathcal{H})$ 
be a discrete POVM.
Then the following conditions are equivalent:
\begin{enumerate}[(i)]
\item
$A$ is pairwise linearly independent;
\item
$A$ is non-vanishing and $\simeq_w$-minimal sufficient.
\item
$A$ is non-vanishing and $\simeq_r$-minimal sufficient.
\end{enumerate}
\end{theo}
\begin{proof}
We first show (i)$\Rightarrow  $(ii).
Assume that $A$ is a pairwise linearly independent POVM.
Then
if we consider the mapping $S$ in the proof of Theorem~\ref{theo:discexuni},
$S$ is an injection and the POVM $\bar{A}$ induced by $S$ is strictly isomorphic to $A$.
Since $\bar{A}$ is $\simeq_w$-minimal sufficient, $A$ is also $\simeq_w$-minimal sufficient.
Thus we have shown (i)$\Rightarrow$(ii).

(ii)$\Rightarrow$(iii) is obvious.

If $A$ is 
non-vanishing and $\simeq_r$-minimal sufficient,
then $\bar{A}$ induced by the statistic
$S$ is strictly isomorphic to $A$ 
by the uniqueness of the minimal sufficient POVM.
Since $\bar{A}$ is pairwise linearly independent,
$A$ is also pairwise linearly independent.
\end{proof}

In Ref.~\onlinecite{martens1990}
it is shown that
for each discrete POVM $A$ there exists a pairwise linearly independent POVM
fuzzy equivalent to $A$
and such a POVM is unique up to strict isomorphism.
This assertion
is a direct corollary of
our Theorem~\ref{theo:discexuni}
and Theorem~\ref{theo:discequi}.

\section{Information conservation condition}
\label{sec:cons}
In this brief section, we consider
information conservation conditions 
proposed by the author~\cite{PhysRevA.91.032110,:/content/aip/journal/jmp/56/9/10.1063/1.4931625}.

Let 
$(\Omega_1 , \mathscr{B}_1)$
be a measurable space.
A completely positive (CP) instrument~\cite{content/aip/journal/jmp/25/1/10.1063/1.526000}
$\mathcal{I}^1_{\cdot} (\cdot)$
(in the Heisenberg picture)
with the outcome space $(\Omega_1 , \mathscr{B}_1)$ is a mapping
\[
	\mathcal{I}^1_{\cdot} (\cdot) \colon
	\mathscr{B}_1 \times \mathcal{L}(\mathcal{H})
	\ni
	(E_1 , a)
	\mapsto
	\mathcal{I}^1_{E_1} (a)
	\in
	\mathcal{L}(\mathcal{H})
\]
such that
\begin{enumerate}[(i)]
\item
for each countable disjoint $\{ E_j \} \subset \mathscr{B}_1$
and each $\rho \in \stsp$ and $a \in \mathcal{L} (\mathcal{H})$,
$
	\tr[ \rho  \mathcal{I}_{\cup_j E_j }^1 (a) ]
	=
	\sum_j
	\tr[ \rho  \mathcal{I}_{ E_j }^1 (a) ] ;
$
\item
$\mathcal{I}_{\Omega_1} (I) =I$;
\item
$\mathcal{I}^1_{E} (\cdot)$ is a normal CP linear map for every $E \in \mathscr{B}_1 .$
\end{enumerate}
A CP instrument simultaneously 
describes the probability distribution of the outcome of a quantum measurement process
and the state change due to the measurement.

Let $\mathcal{I}^1_\cdot (\cdot)$
be a CP instrument with a standard Borel outcome space
$(\Omega_1 , \mathscr{B}_1)$
and let 
$(\Omega_2 , \mathscr{B}_2 , A_2)$
be a standard Borel POVM.
A composition~\cite{davies1976quantum,:/content/aip/journal/jmp/56/9/10.1063/1.4931625} 
$\mathcal{I}^1 \ast A_2$ is a unique POVM
with the product outcome space 
$(\Omega_1 \times \Omega_2 , \mathscr{B}_1 \times \mathscr{B}_2)$
such that 
\[
	( \mathcal{I}^1 \ast A_2 )
	(E_1 \times E_2)
	=
	\mathcal{I}^1_{E_1} (A_2 (E_2))
\]
for each $E_1 \in \mathscr{B}_1$
and
$E_2 \in \mathscr{B}_2$.
The composition corresponds to the joint successive measurement process of $\mathcal{I}^1_\cdot (\cdot)$
followed by $A_2 .$

For a given CP instrument $\mathcal{I}^1_\cdot (\cdot)$ with a standard Borel outcome space
and a given standard Borel POVM $A_2$,
we consider the following two conditions.
\begin{enumerate}
\item
\label{en:cond1}
There exists a sufficient statistic $\tilde{x} \colon \Omega_1 \times \Omega_2 \to \Omega_2$
such that $(\mathcal{I}^1 \ast A_2)_{\tilde{x}} = A_2$.
\item
\label{en:cond2}
$
	\mathcal{I}^1 \ast A_2 \simeq_r A_2.
$
\end{enumerate}
In Ref.~\onlinecite{PhysRevA.91.032110},
the author derived the condition~\ref{en:cond1} as a sufficient condition 
for a so called relative-entropy conservation law.
In Ref.~\onlinecite{:/content/aip/journal/jmp/56/9/10.1063/1.4931625},
the author reformulated this condition in the form of \ref{en:cond2}
and called it an information conservation condition,
while the condition~\ref{en:cond1} is a sufficient condition but not a necessary one 
for the condition \ref{en:cond2}.
Noting that
the condition~\ref{en:cond2} is invariant under replacing 
$A_2$ with another regularly fuzzy equivalent standard Borel POVM~\cite{:/content/aip/journal/jmp/56/9/10.1063/1.4931625},
this discrepancy can be resolved by taking $A_2$ an equivalent $\simeq_w$-minimal sufficient POVM,
which is always possible due to Theorem~\ref{theo:exunique},
and in this sense
the two conditions are essentially equivalent.

\begin{acknowledgments}
The author acknowledges supports
by Japan Society for the
Promotion of Science (KAKENHI Grant No. 269905).
He also would like to thank 
Takayuki Miyadera (Kyoto University)
and 
Erkka Haapasalo
(University of Turku)
for helpful discussions and comments.
\end{acknowledgments}

\appendix
\section{Inequivalence of $\simeq_r $ and $\simeq_w$}
\label{sec:app1}
In this appendix, we construct a pair of POVMs 
that are weakly fuzzy equivalent but not regularly fuzzy equivalent.

Let $(\Omega_A , \mathscr{B}_A , A)$ be a POVM.
The completion of $(\Omega_A , \mathscr{B}_A , A)$ is 
a POVM $(\Omega_A , \bar{\mathscr{B}}_A , \bar{A})$
defined by
\begin{gather*}
	\mathscr{N}_A
	:=
	\Set{
	N \subset \Omega_A
	|
	\exists N^\prime \in \mathscr{B}_A \text{ s.t.} \,  N \subset N^\prime  \text{ and } A (N^\prime) = O
	} ,
	\\
	\bar{ \mathscr{B} }_A
	:=
	\Set{ 
	E \subset \Omega_A 
	|   
	\exists F \in \mathscr{B}_A 
	\text{ s.t.} \,
	E \triangle F \in \mathscr{N}_A 
	} ,
	\\
	\bar{A}(E) := A (F)
	,
	\quad
	(E \in \bar{\mathscr{B}}_A , F \in \mathscr{B}_A , E \triangle F \in \mathscr{N}_A).
\end{gather*}
Here $E \triangle F := ( E \setminus F )  \cup ( F \setminus E)$
is the symmetric difference of sets.
An element of $\mathscr{N}_A$ is called an $A$-null set.

\begin{lemm}
\label{lemm:completion}
Let $(\Omega_A , \mathscr{B}_A , A)$ be a POVM
and let $(\Omega_A , \bar{\mathscr{B}}_A , \bar{A})$ be the completion of
$(\Omega_A , \mathscr{B}_A , A).$
Then $A \simeq_w \bar{A} .$
\end{lemm}
\begin{proof}
Since $A(E) = \bar{A} (E) = \int_{\Omega_A} \chi_E (x) d \bar{A}(x)$
for each $E \in \mathscr{B}_A ,$
$A \preceq_r \bar{A}$ holds.
For each $E \in \bar{\mathscr{B}}_A $
we take $F \in \mathscr{B}_A$ such that $E \triangle F \in \mathscr{N}_A$
and define $\kappa (E|x) := \chi_F(x) .$
Then $\kappa (\cdot| \cdot)$ is an $A$-weak Markov kernel such that
$\bar{A} (E) = \int_{\Omega_A} \kappa (E|x) d A(x) $
for each $E \in \bar{\mathscr{B}}_A .$
Thus we have $\bar{A} \preceq_w A,$
and the assertion holds.
\end{proof}

Let $\mu$ be the usual Lebesgue measure on a unit interval $I := [0,1] ,$
i.e. $\mu$ is the unique measure defined on the $\sigma$-algebra
$\mathscr{B} (I)$ of $I$
generated by open subsets of $I$
such that
$\mu ( [a,b]  ) = b-a$
for each $0 \leq a < b \leq 1 .$
As the system Hilbert space $\mathcal{H},$
we consider the set of square-integrable $\mathscr{B} (I)$-measurable functions
$L^2 (I , \mathscr{B}(I) , \mu)$
in which $\mu\text{-}\mathrm{a.e.}$ equal functions are identified.
We define a projection-valued measure (PVM) $(I , \mathscr{B}(I) , A)$ by
\[
	(A(E) f) (x) := \chi_E (x) f(x)
\]
for each $E \in \mathscr{B}(I)$ and $f \in \mathcal{H} .$
We denote the completion of $(I , \mathscr{B}(I) , A)$ by
$(I , \bar{\mathscr{B}} (I ) , \bar{A}) .$
Since the class of $A$-null sets and that of $\mu$-null sets coincide,
$\bar{\mathscr{B}} (I)$ is the class of Lebesgue measurable subsets of $I .$
From Lemma~\ref{lemm:completion},
we have $A \simeq_w \bar{A}$.

Now we prove that $\bar{A} \not\preceq_r A ,$
from which we immediately obtain the desired relation $A \not\simeq_r \bar{A} .$
Suppose that $\bar{A} \preceq_r A$ holds.
Then there exists a regular Markov kernel
$\kappa(\cdot| \cdot) \colon \bar{\mathscr{B}} (I) \times I \to [0,1]$
such that 
\[
	\bar{A} (E) 
	=
	\int_{I} \kappa (E|x) dA (x)
\]
for each $E \in \bar{\mathscr{B}} (I) .$
Thus we have
\begin{equation}
	\int_I \chi_E (x) d A(x)
	=
	A(E) = \bar{A} (E)
	=
	\int_I \kappa (E|x) d A(x)
	\label{eq:chikappa} 
\end{equation}
for each $E \in \mathscr{B}(I) .$
From Remark~5 of Ref.~\onlinecite{Dorofeev1997349}, 
Eq.~\eqref{eq:chikappa} implies that
$\chi_E(x) = \kappa (E|x)$ 
for $\mu$-a.e. $x \in I. $
Therefore there exists a $\mu$-null set $N \in \mathscr{B}(I)$
such that
\begin{equation}
	\kappa ([0 , r]  |x) = \chi_{[0,r]} (x)
	\quad
	(\forall r \in I \cap \mathbb{Q} , \forall x \in I \setminus  N),
	\label{eq:rational}
\end{equation}
where $\mathbb{Q}$ is the set of rational numbers.
Noting that $\kappa (\cdot |x)$ is a probability measure 
for each $x \in I ,$
Eq.~\eqref{eq:rational} indicates that 
$\kappa (E|x) = \chi_E (x)$ for each $x \in I \setminus N  $
and $E \in \mathscr{B}(I) .$
Thus 
we have 
$\kappa (  I \setminus  \{ x \}  | x ) = 0$
for each $x \in I \setminus N $
and this implies that 
$\kappa (E | x)  = \chi_E (x)$
for each $x \in I \setminus N$ and $E \in \bar{\mathscr{B}} (I) .$ 
If there exists a set $E $ such that
\begin{equation}
	E \subset I \setminus N,
	\quad
	E \in \bar{\mathscr{B}} (I) \setminus \mathscr{B}(I) ,
	\label{eq:Econd}
\end{equation}
then
we have $ \chi_E (x) = \kappa (E|x) \chi_{I \setminus N} (x) $,
which contradicts the $\mathscr{B}(I)$-nonmeasurability of 
$\chi_E (x) .$

Now we show the existence of $E$ satisfying~\eqref{eq:Econd}.
Since $( I \setminus N ,  (I \setminus N ) \cap \mathscr{B}(I) )$
is a standard Borel space and the restriction of $\mu$ to $I \setminus N$ is a continuous  measure,
from Theorem~17.41 of Ref.~\onlinecite{kechris1995classical},
there exists a $  (I \setminus N ) \cap  \mathscr{B}(I) / \mathscr{B}(I)$-bimeasurable
bijection $f \colon I\setminus N \to I$
such that 
$\mu (f(E)) = \mu (E) $
for every $E \in (I \setminus N ) \cap \mathscr{B}(I) .$
Since there exists a set $\tilde{E} \in \bar{\mathscr{B}} (I) \setminus \mathscr{B}(I) ,$
$E := f^{-1} (\tilde{E})$ satisfies the condition~\eqref{eq:Econd},
which completes the proof of the assertion.

\end{document}